\documentclass{article}
\usepackage[utf8]{inputenc}
\usepackage{times}

\usepackage{soul}
\usepackage[small]{caption}

\usepackage{appendix}
\usepackage{geometry}
\usepackage{amsmath}
\usepackage{amsfonts}
\usepackage{amsthm}
\usepackage{amstext}
\usepackage{mathrsfs}
\usepackage{amssymb}
\usepackage{graphicx}
\usepackage{subfigure}
\usepackage{algorithm,algorithmicx}
\usepackage[noend]{algpseudocode}
\usepackage{bm}
\usepackage{color}
\usepackage{natbib}
\usepackage{hyperref}
\usepackage{url}

\newtheorem{definition}{Definition}[section]
\newtheorem{lemma}{Lemma}[section]
\newtheorem{theorem}{Theorem}[section]
\newtheorem{corollary}{Corollary}[section]
\newtheorem{proposition}{Proposition}[section]

\newtheorem{example}{Example}[section]
\newcommand{\median}{\mathrm{med}}

\title{Strategyproof Facility Location Mechanisms with Richer Action Spaces}
\author{
Xiang Yan$^{1,2}$ and Yiling Chen$^2$ \\
{$^1$Shanghai Jiao Tong University}\\
{$^2$Harvard University}\\
xiangyan@g.harvard.edu, yiling@seas.harvard.edu
}

\date{}

\begin{document}

\maketitle

\begin{abstract}

We study facility location problems where agents control multiple locations and when reporting their locations can choose to hide some locations (hiding), report some locations more than once (replication) and lie about their locations (manipulation). We fully characterize all facility location mechanisms that are anonymous, efficient, and strategyproof with respect to the richer strategic behavior for this setting. We also provide a characterization with respect to manipulation only. 
This is the first, to the best of our knowledge, characterization result for the strategyproof facility location mechanisms where each agent controls multiple locations.
\end{abstract}

\section{Introduction}

In a classic facility location problem, a social planner chooses to build a facility based on reported locations of agents on a real line. Each agent has one private location and prefers the facility to be built as close to her location as possible. Agents may choose to lie about their locations to influence where the facility is built. It is well-known that choosing the median of reported locations is not only  strategyproof but also socially optimal, resulting in the smallest total distance between the facility and agents' locations. \cite{moulin1980strategy}'s seminal work fully characterizes all strategyproof facility location mechanisms for this setting. 



In many scenarios, for example when each agent represents a community, agents may control more than one locations. The social planner still hopes to select a location to build the facility based on agents' reported locations. The facility location problem where each agent controls multiple locations was first introduced by \cite{procaccia2009approximate}. Choosing the social optimal solution, the median of all reported locations, is no longer strategyproof. \cite{procaccia2009approximate} provided an intuitive, strategyproof mechanism that relabels each reported location by the corresponding agent's most preferred location and then chooses the median of the relabelled locations. 
The characterization of all strategyproof mechanisms remains an open question. 

Moreover, strategyproof mechanisms so far only guard against one type of strategic behavior, that is agents may lie about their locations in reporting (which we call manipulation in this paper). But when agents control more than one locations, they can choose to hide some locations (hiding) or report some locations more than once (replication) to influence the facility location to their benefit. Strategyproof mechanisms need to be robust against these richer strategic actions. The above-mentioned mechanism provided by \cite{procaccia2009approximate} is not strategyproof w.r.t. replication. \cite{hossainsurprising} provided examples showing that facility location mechanisms where each agent controls multiple locations may be strategyproof w.r.t. manipulation but not hiding, or vice versa.

In this paper, we fully characterize strategyproof mechanisms w.r.t. the richer strategic actions for  facility location problems where each agent controls multiple locations.
Intuitively, for each agent reporting multiple locations, one may treat her as an agent controlling a single location, her most preferred location, which is the median of her reported locations.
By doing so, all the strategyproof mechanisms for single location agents should also be strategyproof for multiple location settings.
One natural guess is that these are all strategyproof mechanisms we desired.
This is however not true. We show that there exists other strategyproof mechanisms that depend on not only agents' most preferred locations but also their other reported locations.

To fully characterize all strategyproof facility location mechanisms where each agent controls multiple locations, we first show that if one cannot distinguish which locations are reported by the same agent, referred as settings with non-identifying locations, all strategyproof mechanisms outputs a constant location that is independent of the reported locations.
Then for settings with identifying locations, we show a necessary property for any strategyproof mechanism: for each agent, the mechanism has at most two possible outputs fixing other agents' reports and the agent's most preferred location(s).
Further adding the Pareto efficiency condition, we derive a full characterization for strategyproof mechanisms.
Finally, we compare the result with the characterization for strategyproof mechanisms w.r.t. manipulation only, 
and discuss the group strategyproofness of the mechanisms.

\subsection{Related Work}
\cite{dekel2010incentive} investigated the framework of mechanism design problems in general learning-theoretic settings.
The facility location problem where each agent controls multiple locations is its special case for one-dimensional linear regression, and was first introduced by \cite{procaccia2009approximate}.
Some deterministic and randomize mechanisms are provided by \cite{procaccia2009approximate}, as well as \cite{lu2009tighter}, and they focused on the approximation ratio of these mechanisms, i.e. the total distance between the facility and all agents' locations compared with the optimal solution. 
\cite{hossainsurprising} extended these studies for considering another strategic behavior, hiding, and provided a strategyproof mechanism w.r.t. both manipulation and hiding.
On the other hand, \cite{conitzer2008anonymity,todo2011false} studied the strategyproofness w.r.t. false-name report, i.e. agents are able to create multiple anonymous identifiers.
It is similar to our replication when the identifier of the agent reporting each locations is unknown.

The facility location problems is also a framework for studying ``single-peak" preference.
As agents with such kind of preference are commonly seen in political decision making like voting \citep{moulin1980strategy}, location on networks \citep{schummer2002strategy}, and resource allocation \citep{guo2010strategy}.
An extension for the preference structure is the ``single-plateau" preference, where the most preferred locations of agents become intervals.
\cite{moulin1984generalized,berga1998strategy} provided corresponding characterizations of strategyproof mechanisms for more general social choice settings.
Following them, the characterizations were extended to high dimensional Euclidean space \citep{border1983straightforward,barbera1993generalized,barbera1998strategy} and convex spaces \citep{tang2018characterization}.
These works considered manipulation as the only strategic behavior, and agents are assumed to only report their ``peak" or ``plateau" preferences. 
In comparison, our work focuses on facility location settings where agents have ``single-plateau" preferences and richer actions spaces. 

\section{Our Model}

Following a brief explanation of notations, we introduce the facility location problem where each agent controls multiple locations, our richer strategic considerations, and desired properties of mechanisms for this problem.  
\paragraph{Notations.}
Let $[k]\overset{\triangle}{=}\{1,\dots,k\}$ be the set of first $k$ natural numbers.
Given $k$ real numbers $t_1,\dots,t_k\in\mathbb{R}\cup\{-\infty,+\infty\}$ and $t_1\leq\dots\leq t_k$, let $\median (t_1,\dots,t_k)$ be their median.
When $k$ is odd, $\median (t_1,\dots,t_k)=t_{(k+1)/2}$, and when $k$ is even, $\median (t_1,\dots,t_k)=[t_{k/2},t_{k/2+1}]$, which is a number if $t_{k/2}=t_{k/2+1}$ and an interval otherwise.
Given a set $S$ of real numbers, possibly with some identical elements, $|S|$ denotes the size of $S$, and  $\median(S)$ is the median of all numbers in $S$.

\paragraph{Facility Location Problems.}
There is a set of agents, $[n]$.
Each agent $i\in [n]$ controls a set of locations, $\bar{D_i}=\{\bar{y}_{i,j}\}_{j\in[|\bar{D_i}|]}$, $\bar{y}_{i,j}\in\mathbb{R}$, on the real line. Each agent $i\in [n]$ is asked to report her set of locations to the principal. Let $D_i = \{y_{i,j}\}_{j\in [|D_i|]}$ denote agent $i$'s set of reported locations,with $y_{i,j}\in\mathbb{R}$. $D_i$ can differ from $\bar{D_i}$ in both size and values. We use $\mathcal{D} = \{D_i \}_{i\in [n]}$ to denote the set of locations reported by all agents and $N=\sum_{i\in [n]}|D_i|$ to represent the total number of reported locations. 
The principal seeks a mechanism $\pi:\mathbb{R}^N \rightarrow \mathbb{R}$ such that $\pi(\mathcal{D})$ is a proper location for building a facility that will be used by all agents. 
Each agent incurs a loss 
\begin{equation}\label{eqn:loss}
    l(\pi(\mathcal{D}),\bar{D_i})= \sum\limits_{\bar{y}_{i,j}\in \bar{D_i}}|\pi(\mathcal{D})-\bar{y}_{i,j}|
\end{equation}
for using a facility located at $\pi(\mathcal{D})$.  
This loss function means that an agent's loss is minimized when the facility locates within an interval (the median of the agent's locations) and strictly increases as the facility moves away from the interval on either side. The following proposition formalizes this property. 
%
%
\begin{proposition}
Loss function (\ref{eqn:loss}) represents ``single-plateau" preferences. Let $[y_i^0,y_i^1]=\median(\bar{D_i})$, where $y_i^0\leq y_i^1$,  represent the median interval of agent $i$'s locations. Then  
\begin{itemize}
\item $\forall y \in [y_i^0,y_i^1]$, $l(y,\bar{D_i}) = \min_{y'\in\mathbb{R}}l(y',\bar{D_i})$, 
\item $\forall s_1,s_2\in\mathbb{R}$ satisfying $s_1<s_2< y_i^0$, $l(s_1,\bar{D_i})>l(s_2,\bar{D_i})>l(y_i^0,\bar{D_i})$, and 
\item $\forall h_1,h_2\in\mathbb{R}$ satisfying $h_1>h_2> y_i^1$, $l(h_1,\bar{D_i})>l(h_2,\bar{D_i})>l(y_i^1,\bar{D_i})$.
\end{itemize}
\end{proposition}

\paragraph{Strategic considerations.}
Agents want to report $D_i$ to minimize their loss. We allow three types of agent strategic behavior:
\begin{itemize}
    \item \textbf{Manipulation.} Each agent $i\in [n]$ may report different value for each of her controlled locations. This is the strategic behavior usually considered in the literature. 
    \item \textbf{Replication.} Each agent $i$ may report one or more of her locations for more than once. Note that even if an agent reports her controlled locations without replication, it is possible that some locations appear for more than once.
    \item \textbf{Hiding.} Each agent $i$ may choose to no report some of her controlled locations.
\end{itemize}
The combination of these three types of strategic behavior allows agents to report a set of locations with any size and any value.

\paragraph{Desirable properties of mechanisms.}
The principle hopes to find a mechanism $\pi$ that discourages strategic behavior of agents. We define three notions of strategyproofness that we'll consider. The notions of strategyproofness are with respect to one or more of the three types of strategic behavior. 
\begin{definition}
A mechanism $\pi(\mathcal{D})$ is \textbf{strategyproof} w.r.t. some set of strategic behavior if no agent can achieve less loss by deviating from truthfully reporting to a strategic action in the set, regardless of the reports of the other agents.
\end{definition}

\begin{definition}
A mechanism $\pi(\mathcal{D})$ is \textbf{group strategyproof} w.r.t. a set of strategic behavior if no coalition of agents can simultaneously adopt strategic actions in the set such that every agent in the coalition is strictly better off, regardless of the reports of the other agents.
\end{definition}

\begin{definition}
A mechanism $\pi(\mathcal{D})$ is \textbf{strong group strategyproof} w.r.t. a set of strategic behavior if no coalition of agents can simultaneously adopt strategic actions in the set such that no agent in the coalition is strictly worse off and some agent in the coalition is strictly better off, regardless of the reports of the other agents.
\end{definition}

Note that a group strategyproof mechanism must be strategyproof and a strong group strategyproof mechanism must be group strategyproof, but not vice versa.
We will discuss the corresponding difference in Section \ref{sec:gsp}.

In addition to strategyproofness, two other properties are also desirable for facility location mechanisms.

\begin{definition}
A facility location mechanism $\pi(\mathcal{D})$ is \textbf{anonymous} if its output is symmetric w.r.t. all agents.
\end{definition}

\begin{definition}
A facility location mechanism $\pi(\mathcal{D})$ is \textbf{efficient} if its output is Pareto optimal to all agents, i.e. there does not exist another location that is strictly better for at least one agent and not worse for all other agents.
\end{definition}

When each agent only controls a single location, denoted by $y_i$ for $i\in [n]$, \cite{moulin1980strategy} characterizes that all facility location mechanisms that are anonymous and strategyproof w.r.t. manipulation are of the form
\begin{equation}\label{eqn:Moulin1}
\pi(\mathcal{D}) = \median(y_1,\dots,y_n,\alpha_1,\dots,\alpha_{n+1}), \forall y
\end{equation}
where $\alpha_1,\dots,\alpha_{n+1} \in \mathbb{R}\cup\{-\infty, +\infty\}$ are constants;
all facility location mechanisms that are anonymous, efficient, and strategyproof w.r.t. manipulation take the form 
\begin{equation}\label{eqn:Moulin2}
\pi(\mathcal{D}) = \median(y_1,\dots,y_n,\alpha_1,\dots,\alpha_{n-1}), \forall y
\end{equation}
where $\alpha_1,\dots,\alpha_{n-1} \in \mathbb{R}\cup\{-\infty, +\infty\}$ are constants.

\section{Strategyproof Facility Location Mechanisms for Non-identifying Locations}\label{sec:unidentify}

We first show that when each agent can control multiple locations, being able to identify which locations are reported by the same agent is necessary for developing non-trivial strategyproof mechanisms. We use the term non-identifying locations to represent the case when one cannot tell which agent reports which location, or more formally, for any reported location $y_i^j\in\mathcal{D}$, 
one cannot distinguish the agent $i$ who has reported the location.  
%
We show below that with non-identifying locations, all strategyproof mechanisms must be the trivial constant mechanisms.

\begin{theorem}\label{thm:unidentified_constant}
For facility location problems where each agent controls multiple locations, if the reported locations are non-identifying, 
then any mechanism that is strategyproof w.r.t. manipulation must output a constant location.
\end{theorem}

\begin{proof}
For convenience, denote the reported data set $\mathcal{D}$ as $\{y_j\}_{j\in [N]}$.
Since the reported locations are non-identifying, it is possible that $|D_i|=1$ for each $i$.
Then according to \cite{moulin1980strategy}, any strategyproof mechanism has the form 
$$\pi(\mathcal{D}) = \median(y_1,\dots,y_N,\alpha_1,\dots,\alpha_{N+1}), \forall y.$$
We next prove that $\alpha_1=\dots=\alpha_{N+1}=\alpha$ for some $\alpha\in\mathbb{R}\cup\{-\infty,+\infty\}$, which means the mechanism always returns a constant $\alpha$.
Suppose otherwise, w.l.o.g, let $\alpha_1 < \alpha_2 \leq \alpha_3\dots \leq \alpha_{N+1}$.
Then we can construct an example where an agent could achieve smaller loss through manipulation.
Consider $\bar{D_1}= \{y_1=\alpha_1\}$ and $\bar{D_2} = \{y_2,\dots,y_{N}\}$ as the real locations controlled by agent $1$ and $2$ respectively, where $y_N = \alpha_2$ and $y_j=(\alpha_1+\alpha_2)/2$ for $j=2,\dots,N-1$.
Then truthfully reporting results in $\pi(\bar{D_1},\bar{D_2}) = \alpha_2$ and agent $2$ suffers a loss of $(N-2)(\alpha_2-\alpha_1)/2$.
If agent $2$ misreports his locations by manipulating $y_{N}$ to $y_{N}'=(\alpha_1+\alpha_2)/2$, then the mechanism will output $(\alpha_1+\alpha_2)/2$ and her loss becomes $(\alpha_2-\alpha_1)/2$, which is strictly smaller if $N>3$.
\end{proof}

\begin{corollary}
For facility location problems where each agent controls multiple locations, if the reported locations are non-identifying, 
then any mechanism which is strategyproof w.r.t. manipulation, replication, and hiding, must output a constant location.
\end{corollary}

Thus, in the rest of the paper, we focus on facility location problems with identifying locations. We note that identifying locations do not conflict with anonymity. Anonymity means that a mechanism's outcome is not affected by relabelling of the agents, while identifying locations only require that one knows which locations are reported by the same agent and the labels of the agents are not important. 

\section{Strategyproof Facility Location Mechanisms for Identifying Locations}\label{sec:sp}

In this section, we characterize mechanisms that are anonymous, efficient and strategyproof w.r.t. all three types of strategic behavior. We show that these mechanisms take the form of Moulin's characterization (\ref{eqn:Moulin2}) where the median of each agent's reported locations is used as the representative location for the agent. We further show in Section \ref{sec:mani} that the family of anonymous, efficient and strategyproof mechanisms remain the same even if only manipulation is considered, when agents control the same number of locations. 

We develop our results by first characterizing mechanisms that are anonymous and strategyproof w.r.t. all three types of strategic behavior. Lemma \ref{lem:most_two_output} shows that from agent $i$'s perspective, fixing other agents' reports, any anonymous and strategyproof mechanism can have at most two different outputs if the median of the agent's reported locations, $[y_i^0,y_i^1]=\median(D_i)$, is unchanged.

\begin{lemma}\label{lem:most_two_output}
For facility location problems where each agent controls multiple locations, if a mechanism $\pi$ is anonymous and strategyproof w.r.t. manipulation, replication, and hiding, then from any agent $i$'s perspective, for any $y_i^0\leq y_i^1\in\mathbb{R}$ and all $D_i$ with $[y_i^0,y_i^1]=\median(D_i)$, the output of the mechanism is either $\pi(\mathcal{D})\in [y_i^0,y_i^1]$ or \begin{equation}\label{eqn:charac_fl}
    \pi(D_i,D_{-i})=\left\{
    \begin{aligned}
    &s_i ~~~~~~~~~\mathrm{if}~l(s_i,D_i) < l(h_i,D_i);\\
    &h_i ~~~~~~~~~\mathrm{if}~l(s_i,D_i) > l(h_i,D_i);\\
    &s_i ~\mathrm{or}~ h_i ~~~\mathrm{otherwise} 
    \end{aligned}
    \right.
\end{equation}
for some $s_i,h_i\in\mathbb{R}$ that satisfies either $s_i< y_i^0\leq y_i^1< h_i$ or $s_i=h_i$. The values of $s_i$ and $h_i$ depend on the reports of other agents $D_{-i}$ as well as $y_i^0$ and $y_i^1$. 

\end{lemma}

\begin{proof}
If a mechanism outputs the optimal location based an agent's reported locations or a constant location, the agent has no incentive to misreport.
Thus, we try to characterize strategyproof mechanisms beyond these two trivial types.

For each $i\in [n]$ and $D_{-i}$ fixed, let $[y_i^0,y_i^1]=\median(D_i)$ be the interval that any location inside is optimal for agent $i$'s real controlled locations.
Through manipulation, replication and hiding, agent $i$ is able to misreport any possible location sets.
Currently we focus on a special group of misreports, that is agent $i$ does not change the median value (either one or two value) of her controlled locations but may misreport any others, i.e. all $D_i'$ with $[y_i^0,y_i^1]=\median(D_i')$.

Suppose there are at least three different outputs $\pi(\{D_i',D_{-i}\})$ for all such $D_i'$ which do not belong to the interval $[y_i^0,y_i^1]$, w.l.o.g. let $\pi(\{D_i^1,D_{-i}\}) < \pi(\{D_i^2,D_{-i}\}) < \pi(\{D_i^3,D_{-i}\})$, $\pi(\{D_i^j,D_{-i}\})\notin [y_i^0,y_i^1]$ for $j=1,2,3$.
If $y_i^1 < \pi(\{D_i^2,D_{-i}\})$, then when other agents report $D_{-i}$, agent $i$ with real controlled locations $\bar{D_i} = D_i^3$ will misreport $D_i^2$ to obtain a smaller loss.
Otherwise, $y_i^0 > \pi(\{D_i^2,D_{-i}\})$, then when other agents report $D_{-i}$, agent $i$ with real controlled locations $\bar{D_i} = D_i^1$ will misreport $D_i^2$ to obtain a smaller loss.
This is a contradiction, which means there are at most two different outputs for all such $D_i'$ which do not belong to the interval $[y_i^0,y_i^1]$, denoted by $s_i \leq h_i$. 
Furthermore, if there exists some $D_i^4$ such that $\pi(\{D_i^4,D_{-i}\}) \in [y_i^0,y_i^1]$, then for all $D_i'$, $\pi(\{D_i',D_{-i}\}) \in [y_i^0,y_i^1]$.
This means for all $D_i$, either $\pi(D_i,D_{-i})\in [y_i^0,y_i^1]$ or $\pi(D_i,D_{-i})\in \{s_i,h_i\}$.

We further consider the case that $s_i\neq h_i$.
Let $\pi(\{D_i^1,D_{-i}\}) =s_i$ and $\pi(\{D_i^2,D_{-i}\})=h_i$.
If $y_i^1 < s_i$, then when other agents report $D_{-i}$, agent $i$ with real controlled locations $\bar{D_i}=D_i^2$ will misreport $D_i^1$ to obtain smaller loss.
Similarly, if $y_i^0 > h_i$, then when other agents report $D_{-i}$, agent $i$ with real controlled locations $\bar{D_i}=D_i^1$ will misreport $D_i^2$ to obtain smaller loss.
Thus, we know $s_i<y_i^0\leq y_i^1<h_i$, and Eqn. (\ref{eqn:charac_fl}) is straightforward.
\end{proof}

The following example shows that a strategyproof mechanism indeed can have two different outputs, i.e. satisfying Eqn. (\ref{eqn:charac_fl}) for $s_i\neq h_i$. 

\begin{example}
For simplicity, suppose $\median(D_i)$ is a unique number for each $i\in[n]$.
For given constants $t_1<t_2\in\mathbb{R}$, let $P=|\{i| i\in[n],l(t_1,D_i)\leq l(t_2,D_i)\}|$ be the number of agents who prefer $t_1$ to $t_2$, and $Q=|\{i| i\in[n],l(t_1,D_i)> l(t_2,D_i)\}|$ be the number of agents who prefer $t_2$ to $t_1$.
Then the following mechanism is strategyproof:
\begin{equation}
    \pi(D_i,D_{-i})=\left\{
    \begin{aligned}
    &t^* ~~~\mathrm{if}~t^* > t_2 ~\mathrm{or}~t^* < t_1;\\
    &t_1 ~~~\mathrm{if}~t_1 < t^* < t_2 ~\mathrm{and}~P \geq Q;\\
    &t_2 ~~~\mathrm{if}~t_1 < t^* < t_2 ~\mathrm{and}~P < Q;\\
    \end{aligned}
    \right.
\end{equation}
where $t^* = \median(\median(D_1),\dots,\median(D_n),\alpha_1,\dots,\alpha_{n+1})$.
\end{example}

Notice that the mechanism in this example is not efficient. We formally prove this observation in the next lemma.

\begin{lemma}\label{lem:sp_eff_only_peak}
For facility location problems where each agent controls multiple locations, if a mechanism $\pi$ is anonymous, efficient and strategyproof w.r.t. manipulation, replication, and hiding, then for any $i\in [n]$, $D_{-i}$, $y_i^0\leq y_i^1\in\mathbb{R}$, there do not exist $s_i,h_i\in\mathbb{R}$ such that $s_i< y_i^0\leq y_i^1< h_i$, and $\pi$ satisfies Eqn. (\ref{eqn:charac_fl}) for all $D_i$ with $[y_i^0,y_i^1]=\median(D_i)$.
\end{lemma}


The proof of Lemma \ref{lem:sp_eff_only_peak} makes use of Lemma \ref{lem:even}.

\begin{lemma}\label{lem:even}
For facility location problems where each agent controls multiple locations, suppose a mechanism $\pi$ is anonymous, efficient and strategyproof w.r.t. manipulation, replication, and hiding.
If there exist some $i\in [n]$, $D_{-i}$, $y_i^0\leq y_i^1\in\mathbb{R}$, $s_i,h_i\in\mathbb{R}$ such that $s_i< y_i^0\leq y_i^1< h_i$, and $\pi$ satisfies Eqn. (\ref{eqn:charac_fl}) for all $D_i$ with $[y_i^0,y_i^1]=\median(D_i)$, denote two index sets $\mathcal{I}_1=\{j|[y_j^0,y_j^1]=\median(D_j), y_j^0\leq s_i,y_j^1<h_i\}$ and $\mathcal{I}_2=\{j|[y_j^0,y_j^1]=\median(D_j), y_j^1\geq h_i,y_j^0>s_i\}$, then $\mathcal{I}_1=\emptyset,\mathcal{I}_2=\emptyset$.
\end{lemma}

\begin{proof}[Proof of Lemma \ref{lem:even}]
If $n=1$, an efficient mechanism should always output the median of the reported locations, which means $\mathcal{I}_1=\mathcal{I}_2=\emptyset$ by definition.

For $n=2$, w.o.l.g., assume there exists some $D_{2}$ with $\median(D_{2})=[y_2^0,y_2^1]$, $y_1^0,y_1^1$, and $s_1< y_1^0\leq y_1^1< h_1$, such that $\pi$ satisfies Eqn. (\ref{eqn:charac_fl}) for all $D_1$ with $[y_1^0,y_1^1]=\median(D_1)$.
This means there exist at least two special location sets $D_1^1$ and $D_1^2$ such that $\pi(D_1^1,D_2)=s_1$ and $\pi(D_1^2,D_2)=h_1$ (for example $D_1^1=\{s_1,y_1^0,y_1^1,y_1^1\}$ and $D_1^2=\{y_1^0,y_1^0,y_1^1,h_1\}$).
If $\mathcal{I}_1 \neq \emptyset$, that is $y_2^0\leq s_i$ and $y_2^1 < h_i$, then $\pi(D_1^2,D_2)=h_1$ is not efficient since any location between $y_1^0$ and $y_1^1$ is a better output for both agent $1$ and $2$.
If $\mathcal{I}_2 \neq \emptyset$, that is $y_2^1\geq h_i$ and $y_2^0 > s_i$, then $\pi(D_1^1,D_2)=s_1$ is not efficient since any location between $y_1^0$ and $y_1^1$ is a better output for both agent $1$ and $2$.
Thus, the result holds for $n=2$.

For $n\geq 3$, assume there exist some $i$, $D_{-i}$, $y_i^0\leq y_i^1$, and $s_i< y_i^0\leq y_i^1< h_i$, such that $\pi$ satisfies Eqn. (\ref{eqn:charac_fl}) for all $D_i$ with $[y_i^0,y_i^1]=\median(D_i)$.
Consider two special location sets $D_i^1,D_i^2$ satisfying $\pi(D_i^1,D_{-i})=s_i$, $\pi(D_i^2,D_{-i})=h_i$.
In the following proof, we shall derive contradictions for $|\mathcal{I}_1|= 1,\dots, n-1$ by induction, and the corresponding analysis for $\mathcal{I}_2$ is similar.

\paragraph{Step 1:} 
Suppose $|\mathcal{I}_1|= 1$.\\
This means there is only one $j\in [n]$, such that $[y_j^0,y_j^1]=\median(D_j)$, $y_j^0\leq s_i$ and $y_j^1 < h_i$.
Denote $D_{-i,-j}$ as the reported location sets by agents other than agent $i$ and $j$, then we rewrite the fact that $\pi(D_i^1,D_j,D_{-i,-j})=s_i$ and $\pi(D_i^2,D_j,D_{-i,-j})=h_i$.

Step 1.1: Moving $D_j$ inside $[s_i,h_i]$ doesn't change $h_i$.\\
Consider agent $j$'s another possible report $D_j'=D_i^1$.
Due to the efficiency, we know $\pi(D_i^1,D_j',D_{-i,-j}) > s_i \geq y_j^0$.
If $\pi(D_i^2,D_j',D_{-i,-j})>h_i$, then when agent $i$ reports $D_i^2$ and other agents report $D_{-i,-j}$, agent $j$ with real locations $\bar{D_j} = D_j'$ will misreport $D_j$ to obtain a smaller loss.
If $\pi(D_i^2,D_j',D_{-i,-j})<h_i$, then when agent $i$ reports $D_i^2$ and other agents report $D_{-i,-j}$, agent $j$ with real locations $\bar{D_j} = D_j$ will misreport $D_j'$ to obtain a smaller loss.
This means $\pi(D_i^2,D_j',D_{-i,-j})=h_i$.

Step 1.2: $h_i$ is the only possible output.\\
Now for fixed $D_j',D_{-i,-j}$, suppose there exists $s_i'<y_i^0$ such that for some $D_i'$ with $[y_i^0,y_i^1]=\median(D_i')$, $\pi(D_i',D_j',D_{-i,-j})=s_i'$.
Let $s^*=\min_{k\in [n]}\{y_k^0|[y_k^0,y_k^1]=\median(D_k), s_i< y_k^0\leq y_k^1 < h_i\}$, and $k$ denote the corresponding agent with $y_k^0 = s^*$.
Due to the efficiency we must have $s_i' \geq s^*$ , which means there exists $k\neq i$ satisfying $y_k^0\leq s_i', y_k^1 < h_i$.
This means we can repeat Step 1.1 w.r.t. $D_k$ and new $[s_i',h_i]$ (for at most $n-2$ times since $k\neq i$ and $k\neq j$) until no such $k$.
We reuse the notation $D'_{-i}$ after this process (either with the repeated part or not) as the reports of agents other than $i$, and we have for all $D_i$ with $[y_i^0,y_i^1]=\median(D_i)$, $\pi(D_i,D_{-i}')=h_i$.

Step 1.3: There exists a beneficial misreport.\\
Here we prove for the case that the repeated process in Step 1.2 does not happen. (For the case that it does, we only need to replace $j$ and $D_i^1$ by the lass $k$ and $D_i'$ in the process.)
Precisely, $\pi(D_i^1,D_j',D_{-i,-j})=h_i$.
However, when agent $j$ reports $D_j'$ and other agents report $D_{-i,-j}$, agent $i$ with real locations $\bar{D_i}=D_i^1$ will misreport $D_j$, which leads to $\pi(D_j,D_j',D_{-i,-j})=\pi(D_j,D_i^1,D_{-i,-j})=\pi(D_i^1,D_j,D_{-i,-j})=s_i$ by the anonymity, to obtain a smaller loss.
This is a contradiction, meaning it is impossible that $|\mathcal{I}|= 1$.

\paragraph{Step 2.}
Suppose we have proved it is impossible that $|\mathcal{I}|= 1,\dots, t$ for some $t\geq 1$, we consider the case $|\mathcal{I}|= t+1$.\\
Let $\hat{j}=\mathop{\arg\min_{j\in\mathcal{I}_1}\{y_j^0|[y_j^0,y_j^1]=\median(D_j)\}}$.
Similarly, denote $D_{-i,-\hat{j}}$ as the reported location sets by agents other than agent $i$ and $\hat{j}$, and we rewrite the fact that $\pi(D_i^1,D_{\hat{j}},D_{-i,-\hat{j}})=s_i$ and $\pi(D_i^2,D_{\hat{j}},D_{-i,-\hat{j}})=h_i$.

Step 2.1: Moving $D_{\hat{j}}$ inside $[s_i,h_i]$ doesn't change $h_i$.\\
Consider agent $\hat{j}$'s another possible report $D_{\hat{j}}'=D_i^1$.
If $\pi(D_i^1,D_{\hat{j}}',D_{-i,-j}) < y_{\hat{j}}^0$, it is not efficient since at least $y_{\hat{j}}^0$ is a better output for all agents $j\in\mathcal{I}_1$ as well as agent $i$ and not worse for other agents.
If $\pi(D_i^2,D_{\hat{j}}',D_{-i,-\hat{j}})>h_i$, then when agent $i$ reports $D_i^2$ and other agents report $D_{-i,-\hat{j}}$, agent $j$ with real locations $\bar{D_{\hat{j}}}=D_{\hat{j}}'$ will misreport $D_{\hat{j}}$ to obtain a smaller loss.
If $y_{\hat{j}}^0\leq \pi(D_i^2,D_{\hat{j}}',D_{-i,-\hat{j}})<h_i$, then when agent $i$ reports $D_i^2$ and other agents report $D_{-i,-\hat{j}}$, agent $\hat{j}$ with real locations $\bar{D_{\hat{j}}}=D_{\hat{j}}$ will misreport $D_{\hat{j}}'$ to obtain a smaller loss.
This means $\pi(D_i^2,D_{\hat{j}}',D_{-i,-\hat{j}})=h_i$.

Step 2.2: There must be another possible output $s_i'<y_i^0$.\\
Now for fixed $D_{\hat{j}}',D_{-i,-\hat{j}}$, if for all $D_i$ with $[y_i^0,y_i^1]=\median(D_i)$, $\pi(D_i,D_{\hat{j}}',D_{-i,-\hat{j}})=h_i$, specifically we have $\pi(D_i^1,D_{\hat{j}}',D_{-i,-\hat{j}})=h_i$.
Then when agent $\hat{j}$ reports $D_{\hat{j}}'$ and other agents report $D_{-i,-\hat{j}}$, agent $i$ with real locations $\bar{D_i}=D_i^1$ will misreport $D_{\hat{j}}$, which leads to $\pi(D_{\hat{j}},D_{\hat{j}}',D_{-i,-\hat{j}})=\pi(D_{\hat{j}},D_i^1,D_{-i,-\hat{j}})=s_i$ by anonymity, to obtain a smaller loss.
This is a contradiction, meaning there must exist $s_i'<y_i^0$ such that for some $D_i'$ with $[y_i^0,y_i^1]=\median(D_i')$, $\pi(D_i',D_{\hat{j}}',D_{-i,-\hat{j}})=s_i'$.

Step 2.3: $\hat{j}\notin\mathcal{I}_1'$.\\
Now for $D_{\hat{j}}',D_{-i,-\hat{j}}$, there still exist $s_i',h_i'$, such that $s_i' < y_i^0\leq y_i^1 < h_i'=h_i$ and $\pi$ satisfies Eqn. (\ref{eqn:charac_fl}) for all $D_i$ satisfying that $|D_i|$ is an even number and $[y_i^0,y_i^1]=\median(D_i)$. 
And for the corresponding new index set $\mathcal{I}_1' = \{j|[y_j^0,y_j^1]=\median(D_j), y_j^0\leq s_i',y_j^1<h_i\}$ (after replacing $D_{\hat{j}}$ by $D_{\hat{j}}'$), we have $\hat{j}\notin \mathcal{I}_1'$ since $\median(D_{\hat{j}}') = [y_i^0, y_i^1]$.

Step 2.4: Repeating previous steps leads to $|\mathcal{I}_1'|=t$.\\
Finally, if $|\mathcal{I}_1'| > t$, we can repeatedly replace the report of the agent whose location set has the leftmost median among those agents in $\mathcal{I}_1'$ in the same way until $|\mathcal{I}_1'| \leq t$, and the same analysis still holds.
Since such an index set satisfies $|\mathcal{I}_1'| \leq n-1$, after at most $n-t-1$ times of such replacing, we must have $|\mathcal{I}_1'| = t$.
However, we have already proved that it is impossible for $|\mathcal{I}_1| = t$.
By induction, we provide contradictions for $|\mathcal{I}|= 1,\dots, n-1$, and this completes the proof of $n\geq 3$.
\end{proof}

\begin{proof}[Proof of Lemma \ref{lem:sp_eff_only_peak}]
Assume there exist some $i$, $D_{-i}$, $y_i^0\leq y_i^1$, and $s_i< y_i^0\leq y_i^1< h_i$, such that $\pi$ satisfies Eqn. (\ref{eqn:charac_fl}) for all $D_i$ with $[y_i^0,y_i^1]=\median(D_i)$.
According to Lemma \ref{lem:even}, we know for any $j\in [n]$, let $[y_j^0,y_j^1]=\median(D_j)$, either $y_j^0\leq s_i < h_i \leq y_j^1$, or $s_i< y_j^0\leq y_j^1 < h_i$.
Let $\mathcal{I}=\{j|[y_j^0,y_j^1]=\median(D_j), s_i< y_j^0\leq y_j^1 < h_i\}$, $\mathcal{I}\neq \emptyset$ since $i\in\mathcal{I}$.
Let $s^*=\min_{j\in\mathcal{I}}\{y_j^0|[y_j^0,y_j^1]=\median(D_j)\}$.
Then for any $D_i$ satisfying $\pi(D_i,D_{-i})=s_i$, it is not efficient since $s^*$ is a better output for all agent $j\in\mathcal{I}$ while not worse for other agent.
This is a contradiction to the efficiency condition.
\end{proof}

Lemma \ref{lem:sp_eff_only_peak} indicates that any mechanism $\pi$ that is anonymous, efficient, and strategyproof, must only depend on the optimal location for each agent.
In other words, agents only need to report their most preferred locations, i.e. $\median(D_i)$ for $i\in[n]$.
With \cite{moulin1980strategy}'s results, we have the complete characterization.

\begin{theorem}\label{thm:sp_charac}
For facility location problems where each agent controls multiple locations, a mechanism $\pi$ is anonymous, efficient and strategyproof w.r.t. manipulation, replication, and hiding, if and only if there exist $\alpha_1,\dots,\alpha_{n-1}\in\mathbb{R}\cup\{-\infty,+\infty\}$, $\beta\in[0,1]$, for $\forall D_1, \dots, D_n$,
\begin{equation}\label{eqn:sp}
  \pi(D_1,\dots,D_n)=\median(y_1^*,\dots,y_n^*,\alpha_1,\dots,\alpha_{n-1})  
\end{equation}
where $y_i^*=\beta y_i^0+(1-\beta)y_i^1$ with $[y_i^0,y_i^1]=\median(D_i)$, for $i=1,\dots,n$.
\end{theorem}

Notice that if $\median(D_i)$ is an interval for some agent $i$, any value in the interval can be regarded as agent $i$'s optimal location.
For simplicity, we only include a tie-breaking rule based on an arbitrary constant $\beta$, which is independent of ($\alpha_1,\dots,\alpha_{n-1}$) and guarantees the strategyproofness.
\cite{moulin1984generalized} provided more general characterizations for strategyproof social choice mechanisms where each agent reports an interval as her ``single-plateau" preference, which deal with the tie-breaking rules carefully.

\subsection{Strategyproofness w.r.t. Manipulation Only}\label{sec:mani}
In most previous studies on facility location problems where each agent controls multiple locations, manipulation is considered as the only strategic behavior that agents may take.
Although some strategyproof mechanisms w.r.t. manipulation are discussed, there is no characterization result.
To characterize such strategyproofness, we further assume that each agent control the same number of locations.

\begin{theorem}\label{thm:simply_mani_charac}
For facility location problems where each agent controls same number of multiple locations, a mechanism $\pi$ is anonymous, efficient and strategyproof w.r.t. manipulation, if and only if there exist $\alpha_1,\dots,\alpha_{n-1}\in\mathbb{R}\cup\{-\infty,+\infty\}$, $\beta\in[0,1]$, for $\forall D_1, \dots, D_n$,
$$\pi(D_1,\dots,D_n)=\median(y_1^*,\dots,y_n^*,\alpha_1,\dots,\alpha_{n-1})$$
where $y_i^*=\beta y_i^0+(1-\beta)y_i^1$ with $[y_i^0,y_i^1]=\median(D_i)$, for $i=1,\dots,n$.
\end{theorem}

The proof of Theorem \ref{thm:simply_mani_charac} follows in a similar spirit as that of Theorem \ref{thm:sp_charac}, which can be found in Appendix \ref{app:sp_mani_only}.
Lemma \ref{lem:most_two_output} still holds but the proof of Lemma \ref{lem:sp_eff_only_peak} needs modification because we cannot directly apply Lemma \ref{lem:even}.  

\section{Group Strategyproof Facility Location Mechanisms for Identifying Locations}\label{sec:gsp}

If a mechanism is group strategyproof w.r.t. some set of strategic behavior, then it must be strategyproof w.r.t. the set of strategic behaviors.
This means facility location mechanisms which are anonymous, efficient and group strategyproof should satisfy Eqn. (\ref{eqn:sp}).
We can further show that these strategyproof mechanisms are indeed group strategyproof.

\begin{theorem}\label{thm:gsp_charac}
For facility location problems where each agent controls multiple locations, a mechanism $\pi$ is anonymous, efficient and group strategyproof w.r.t. manipulation, replication, and hiding, if and only if there exist $\alpha_1,\dots,\alpha_{n-1}\in\mathbb{R}\cup\{-\infty,+\infty\}$, $\beta\in[0,1]$, for $\forall D_1, \dots, D_n$,
$$\pi(D_1,\dots,D_n)=\median(y_1^*,\dots,y_n^*,\alpha_1,\dots,\alpha_{n-1})$$
where $y_i^*=\beta y_i^0+(1-\beta)y_i^1$ with $[y_i^0,y_i^1]=\median(D_i)$, for $i=1,\dots,n$.
\end{theorem}

\begin{proof}
Let $D_i$ be any reported locations controlled by agent $i$ with $[y_i^0,y_i^1]=\median(D_i)$ for $i\in [n]$.
For any agent $i$ and any coalition $S$ of agents including $i$, denote their real locations by $D_{\mathcal{S}}=\{D_{j}\}_{j\in\mathcal{S}}$, and the reported locations of other agents by $D_{-\mathcal{S}}$.
W.l.o.g. assume $y_i^1 < \pi(D_{\mathcal{S}},D_{-\mathcal{S}})$.

Consider any misreport by the coalition $D_{\mathcal{S}}'=\{D_{j}'\}_{j\in\mathcal{S}}$ that satisfies $$l(\pi(D_{\mathcal{S}}',D_{-\mathcal{S}}),D_i) < l(\pi(D_{\mathcal{S}},D_{-\mathcal{S}}),D_i).$$
Then $\pi(D_{\mathcal{S}}',D_{-\mathcal{S}}) < \pi(D_{\mathcal{S}},D_{-\mathcal{S}})$.
This means there exists some $j\in\mathcal{S}$ such that $y_j^1\geq\pi(D_{\mathcal{S}},D_{-\mathcal{S}})$, while after misreporting $D_j'\neq D_j$, some value $y_j'\in\median(D_j')$ satisfies $y_j'<\pi(D_{\mathcal{S}},D_{-\mathcal{S}})$.
Then for agent $j$, if $y_j^0\leq \pi(D_{\mathcal{S}},D_{-\mathcal{S}})$, her loss cannot be smaller after such a misreport since her loss is minimized originally.
If $y_j^0 > \pi(D_{\mathcal{S}},D_{-\mathcal{S}})$, then after such a misreport she obtains a bigger loss.
Thus, at least agent $j$ in the coalition is not strictly better off, which completes the proof for group strategyproofness.

\end{proof}

However, when some agent have multiple optimal locations, i.e. some agent $i$'s real locations $D_i$ satisfies $\median(D_i)=[y_i^0,y_i^1]$ and $y_i^0<y_i^1$, then most strategyproof mechanism satisfying Eqn. (\ref{eqn:sp}) is not strong group strategyproof.
Here is a counter example.

\begin{example}
Let $n=3$. For any anonymous, efficient and strategyproof mechanism $\pi(D_1,D_2,D_3)$, let $\alpha_1\leq \alpha_2$ be the corresponding constants.

If $\alpha_1<\alpha_2$, let $\alpha_1 < y_1^0 < y_1^1 < \alpha_2$, $D_1=\{y_i^0,y_i^1\}$, $D_2=\{y_i^0\}$, and $D_3=\{y_i^1\}$ as the real controlled locations of agent $1,2,3$ respectively.
Then $\pi(D_1,D_2,D_3)=y_1^*$, the value choosing from the interval $[y_i^0,y_i^1]$.
If $y_1^*=y_1^0$, then agent $1$ can misreport $D_1'=\{y_1^1,y_1^1\}$, resulting in $\pi(D_1,D_2,D_3)=y_1^1$ which is better for agent $3$ and not worse for agent $1$ herself.
Otherwise, $y_1^0 < y_1^* \leq y_1^1$, then agent $1$ can misreport $D_1'=\{y_1^0,y_1^0\}$, resulting in $\pi(D_1,D_2,D_3)=y_1^0$ which is better for agent $2$ and not worse for agent $1$ herself.

If $\alpha_1=\alpha_2\notin \{-\infty,+\infty\}$, let $y_i^0 < \alpha_1=\alpha_2 < y_i^1$ and $D_1=\{y_i^0,y_i^1\}$.
Similarly if $y_1^*=y_1^0$, when $D_2=\{y_i^1\}$, and $D_3=\{y_i^1\}$, agent $1$ truthfully reporting $D_1$ results in $\pi(D_1,D_2,D_3)=\alpha_1$.
But agent $1$ can misreport $D_1'=\{y_1^1,y_1^1\}$, resulting in $\pi(D_1,D_2,D_3)=y_1^1$ which is better for agent $2$ and $3$, and not worse for agent $1$.
If $y_1^*$ is chosen as any value satisfies $y_1^0< y_1^* \leq y_1^1$, when $D_2=\{y_i^0\}$, and $D_3=\{y_i^0\}$, agent $1$ truthfully reporting $D_1$ results in $\pi(D_1,D_2,D_3)=y_1^*$.
But agent $1$ can misreport $D_1'=\{y_1^0,y_1^0\}$, resulting in $\pi(D_1,D_2,D_3)=y_1^0$ which is better for agent $2$ and $3$, and not worse for agent $1$.
\end{example}

If each agent is further assumed to have unique optimal location, then strategyproof mechanisms are also strong group strategyproof, the proof of which is similar as Theorem \ref{thm:gsp_charac}.

\begin{theorem}\label{thm:sgsp_charac}
For facility location problems where each agent controls multiple locations and has an unique optimal location, a mechanism $\pi$ is anonymous, efficient and strong group strategyproof w.r.t. manipulation, replication, and hiding, if there exist $\alpha_1,\dots,\alpha_{n-1}\in\mathbb{R}\cup\{-\infty,+\infty\}$, for $\forall D_1, \dots, D_n$,
$$\pi(D_1,\dots,D_n)=\median(y_1^*,\dots,y_n^*,\alpha_1,\dots,\alpha_{n-1})$$
where $y_i^*=\median(D_i)$,$i=1,\dots,n$.
\end{theorem}

\section{Future Directions}
We considered richer strategic behavior of agents in facility location problems where each agent controls multiple locations. Facility location problems can be viewed as a single-dimension special case of the strategic linear regression problem, initially introduced by \cite{dekel2010incentive}. Prior work by \cite{dekel2010incentive} and \cite{chen2018strategyproof} have studied linear regression that are strategyproof w.r.t. manipulation. 
However, the regression mechanism proposed in \cite{dekel2010incentive} is not strategyproof w.r.t. replication.
And the GRH mechanisms studied by \cite{chen2018strategyproof}, are not strategyproof w.r.t. to hiding.
It will be interesting to characterize strategyproof linear regression under richer strategic considerations, especially because in practice commercial data sources often find it beneficial to replicate or hide part of their data.

On the other hand, Section \ref{sec:mani} showed that the family of anonymous, efficient and strategyproof mechanisms remains the same even if only manipulation is considered, when agents control the same number of locations. This is somewhat surprising because one may expect that the richer the strategic consideration the smaller the set of strategyproof mechanisms. It will be interesting to fully characterize the family of strategyproof mechanisms w.r.t. manipulation only, with the assumption that each agent controls the same number of locations dropped.  

\bibliographystyle{named}
\bibliography{ArXIv_StrategyproofFacilityLocation}

\appendix
\section{Proof of Theorem \ref{thm:simply_mani_charac}}\label{app:sp_mani_only}

Let $m\geq 1$ be the number of locations each agent controls, i.e. $m=|D_i|, \forall i$.
The proof is separated into two parts due to the different properties of the median operation when $m$ is an odd number and when $m$ is an even number.

\begin{lemma}\label{app_lem:most_two_output}
For facility location problems where each agent controls multiple locations, if a mechanism $\pi$ is anonymous and strategyproof w.r.t. manipulation, then for each $i\in [n]$, $D_{-i}$, 
\begin{itemize}
    \item if $m$ is an odd number, for each $y_i^*\in\mathbb{R}$, there exist $s_i,h_i\in\mathbb{R}$ such that $s_i< y_i^*< h_i$ or $s_i=h_i$ and 
    \begin{equation}\label{app_eqn:charac_fl_odd}
    \pi(D_i,D_{-i})=\left\{
    \begin{aligned}
    &s_i ~~~~~~~~~\mathrm{if}~l(s_i,D_i) < l(h_i,D_i);\\
    &h_i ~~~~~~~~~\mathrm{if}~l(s_i,D_i) > l(h_i,D_i);\\
    &s_i ~\mathrm{or}~ h_i ~~~\mathrm{otherwise} 
    \end{aligned}
    \right.
    \end{equation}
    for all $D_i$ with $y_i^*=\median(D_i)$;
    \item if $m$ is an even number, for each $y_i^0\leq y_i^1\in\mathbb{R}$, there exist $s_i,h_i\in\mathbb{R}$ such that $s_i< y_i^0\leq y_i^1< h_i$ or $s_i=h_i$, and either $\pi(\mathcal{D})\in [y_i^0,y_i^1]$ or 
    \begin{equation}\label{app_eqn:charac_fl_even}
    \pi(D_i,D_{-i})=\left\{
    \begin{aligned}
    &s_i ~~~~~~~~~\mathrm{if}~l(s_i,D_i) < l(h_i,D_i);\\
    &h_i ~~~~~~~~~\mathrm{if}~l(s_i,D_i) > l(h_i,D_i);\\
    &s_i ~\mathrm{or}~ h_i ~~~\mathrm{otherwise} 
    \end{aligned}
    \right.
    \end{equation}
    for all $D_i$ with $[y_i^0,y_i^1]=\median(D_i)$.
\end{itemize}
\end{lemma}

\begin{proof}
If a mechanism outputs the optimal location based an agent's reported locations or a constant location, the agent has no incentive to misreport.
Thus, we try to characterize a strategyproof mechanism beyond these two trivial types.

First for the case that $m$ is an odd number, each agent $i$ has an unique optimal location, denoted by $y_i^*$.
We focus on all $D_i$ with $y_i^*=\median(D_i)$, which are possible misreports from the agent who manipulate its private locations.
For any fixed $D_{-i}$, suppose there are at least three different outputs $\pi(\{D_i',D_{-i}\})$ for all such $D_i'$, w.l.o.g. let $\pi(\{D_i^1,D_{-i}\}) < \pi(\{D_i^2,D_{-i}\}) < \pi(\{D_i^3,D_{-i}\})$.
If $y_i^* \leq \pi(\{D_i^2,D_{-i}\})$, then when other agents report $D_{-i}$, agent $i$ whose real controlled locations $\bar{D_i} = D_i^3$ will misreport $D_i^2$ to obtain a smaller loss.
Otherwise, $y_i^* > \pi(\{D_i^2,D_{-i}\})$, then when other agents report $D_{-i}$, agent $i$ whose real controlled locations $\bar{D_i} = D_i^1$ will misreport $D_i^2$ to obtain a smaller loss.
This is a contradiction, which means there are at most two different outputs for all such $D_i$, denoted by $s_i < h_i$. 
Let $\pi(\{D_i^1,D_{-i}\}) =s_i$ and $\pi(\{D_i^2,D_{-i}\})=h_i$.
(Here we ignore the case $s_i=h_i$ corresponding to a constant output.)
If $y_i^* \leq s_i$, when other agents report $D_{-i}$, agent $i$ whose real controlled locations $\bar{D_i} = D_i^2$ will misreport $D_i^1$.
Similarly, if $y_i^* \geq h_i$, when other agents report $D_{-i}$, agent $i$ whose real controlled locations $\bar{D_i} = D_i^1$ will misreport $D_i^2$.
Thus, we know $s_i<y_i^*<h_i$, and Eqn. (\ref{app_eqn:charac_fl_odd}) is straightforward.
Finally, the case that $m$ is an even number can be analyzed similarly.

\end{proof}




In the next two lemma, corresponding to $m$ is odd and even, we formally prove the case $s_i\neq h_i$ in Eqn. (\ref{app_eqn:charac_fl_odd}) and Eqn. (\ref{app_eqn:charac_fl_even}) will not happen if the mechanism is also efficient. 

\begin{lemma}\label{app_lem:sp_eff_only_peak_odd}
For facility location problems where each agent controls multiple locations, if a mechanism $\pi$ is anonymous, efficient and strategyproof w.r.t. manipulation, and $m$ is an odd number, then for any $i\in [n]$, $D_{-i}$, $y_i^*\in\mathbb{R}$, there do not exist $s_i,h_i\in\mathbb{R}$ such that $s_i< y_i^*< h_i$, and $\pi$ satisfies Eqn. (\ref{app_eqn:charac_fl_odd}) for all $D_i$ with $y_i^*=\median(D_i)$.
\end{lemma}

\begin{proof}
If $n=1$, an efficient mechanism should always output the median of the reported locations.

For $n=2$, w.o.l.g, assume there exists some $D_{2}$ with median denoted by $y_2^*$, $y_1^*$, and $s_1< y_1^*< h_1$, such that $\pi$ satisfies Eqn. (\ref{app_eqn:charac_fl_odd}) for all $D_1$ with $y_1^*=\median(D_1)$.
This means there exist at least two special location sets $D_1^1$ and $D_1^2$ such that $\pi(D_1^1,D_2)=s_1$ and $\pi(D_1^2,D_2)=h_1$ (for example $D_1^1=\{s_1,y_1^*,y_1^*\}$ and $D_1^2=\{y_1^*,y_1^*,h_1\}$).
If $y_2^*\geq y_1^*$, then $\pi(D_1^1,D_2)=s_1$ is not efficient since at least $y_1^*$ is a better output for both agent $1$ and $2$.
Similarly, if $y_2^* < y_1^*$, then $\pi(D_1^2,D_2)=h_1$ is not efficient since at least $y_1^*$ is a better output for both agent $1$ and $2$.
It is a contradiction, so the result holds for $n=2$.

For $n\geq 3$, assume there exist some $i$, $D_{-i}$, $y_i^*$, and $s_i< y_i^*< h_i$, such that $\pi$ satisfies Eqn. (\ref{app_eqn:charac_fl_odd}) for all $D_i$ with $y_i^*=\median(D_i)$.
Consider three special location sets $D_i^1,D_i^2,D_i^3$ satisfying $\pi(D_i^1,D_{-i})=s_i$, $\pi(D_i^2,D_{-i})=h_i$, and $l(s_i,D_i^3)=l(h_i,D_i^3)$.
For example,
$D_i^1=\{s_i,y_i^*,y_i^*\}$, $D_i^2=\{y_i^*,y_i^*,h_i\}$ and 
\begin{equation}
D_i^3=\left\{
\begin{aligned}
    &\{s_i,y_i^*,\frac{1}{2}(3h_i+s_i)-y_i^*\} ~~~\mathrm{if}~y_i^*\geq\frac{1}{2}(h_i+s_i); \nonumber\\ 
    &\{\frac{1}{2}(h_i+3s_i)-y_i^*,y_i^*,h_i\} ~~~\mathrm{if}~y_i^* < \frac{1}{2}(h_i+s_i). \nonumber\\
\end{aligned}
\right.
\end{equation}
W.l.o.g., we assume $\pi(D_i^3,D_{-i})=h_i$.
Define an index set $\mathcal{I}=\{j|y_j^*=\median(D_j), y_j^*\leq s_i\}$.
By definition and the efficiency condition, we know $1\leq |\mathcal{I}| \leq n-1$.
In the following proof, we shall derive contradictions for $|\mathcal{I}|= 1,\dots, n-1$ by induction, thus the result holds for any $n\geq 3$.

If $|\mathcal{I}|= 1$, that is there is only one $j\in [n]$, such that $y_j^*=\median(D_j)$ and $y_j^*\leq s_i$.
Denote $D_{-i,-j}$ as the reported location sets by agents other than agent $i$ and $j$, then we rewrite the fact that $\pi(D_i^1,D_j,D_{-i,-j})=s_i$ and $\pi(D_i^2,D_j,D_{-i,-j})=h_i$.
Consider agent $j$'s another possible report $D_j'=D_i^1$.
Due to the efficiency, we know $\pi(D_i^1,D_j',D_{-i,-j}) > s_i \geq y_j^*$.
If $\pi(D_i^2,D_j',D_{-i,-j})<h_i$, then when agent $i$ reports $D_i^2$ and other agents report $D_{-i,-j}$, agent $j$ with real locations $\bar{D_j}=D_j$ will misreport $D_j'$ to obtain a smaller loss.
If $\pi(D_i^2,D_j',D_{-i,-j})>h_i$, then when agent $i$ reports $D_i^2$ and other agents report $D_{-i,-j}$, agent $j$ with real locations are $\bar{D_j}=D_j'$ will misreport $D_j$ to obtain a smaller loss.
This means $\pi(D_i^2,D_j',D_{-i,-j})=h_i$.
Now for fixed $D_j',D_{-i,-j}$, if there exists $s_i'<y_i^*$ such that for some $D_i'$ with $y_i^*=\median(D_i')$, $\pi(D_i',D_j',D_{-i,-j})=s_i'$, then $\pi(D_i',D_j',D_{-i,-j})$ is not efficient.
This means for all $D_i'$ with $y_i^*=\median(D_i')$, $\pi(D_i',D_j',D_{-i,-j})=h_i$.
Specifically, $\pi(D_i^1,D_j',D_{-i,-j})=h_i$.
However, when agent $j$ reports $D_j'$ and other agents report $D_{-i,-j}$, agent $i$ with real locations are $D_i^1$ will misreport $D_j$, which leads to $\pi(D_j,D_j',D_{-i,-j})=\pi(D_j,D_i^1,D_{-i,-j})=s_i$ by the anonymity, to obtain a smaller loss.
This is a contradiction, meaning it is impossible that $|\mathcal{I}|= 1$.

Now suppose we have proved it is impossible that $|\mathcal{I}|= 1,\dots, t$ for some $t\geq 1$, we consider the case $|\mathcal{I}|= t+1$.
Let $\hat{j}=\mathop{\arg\min_{j\in\mathcal{I}}\{y_j^*|y_j^*=\median(D_j)\}}$, then $\forall j \in [n]$, $y_j^*=\median(D_j)\geq y_{\hat{j}}^*$.
Similarly, denote $D_{-i,-\hat{j}}$ as the reported location sets by agents other than agent $i$ and $\hat{j}$, and we rewrite the fact that $\pi(D_i^1,D_{\hat{j}},D_{-i,-\hat{j}})=s_i$ and $\pi(D_i^2,D_{\hat{j}},D_{-i,-\hat{j}})=h_i$.
Consider agent $\hat{j}$'s another possible report $D_{\hat{j}}'=D_i^1$.
Due to the efficiency, we know $\pi(D_i^1,D_{\hat{j}}',D_{-i,-j})\geq y_{\hat{j}}^*$.
If $\pi(D_i^2,D_{\hat{j}}',D_{-i,-\hat{j}})<h_i$, then when agent $i$ reports $D_i^2$ and other agents report $D_{-i,-\hat{j}}$, agent $\hat{j}$ with real locations $\bar{D_{\hat{j}}}=D_{\hat{j}}$ will misreport $D_{\hat{j}}'$ to obtain a smaller loss.
If $\pi(D_i^2,D_{\hat{j}}',D_{-i,-\hat{j}})>h_i$, then when agent $i$ reports $D_i^2$ and other agents report $D_{-i,-\hat{j}}$, agent $j$ with real locations $\bar{D_{\hat{j}}}=D_{\hat{j}}'$ will misreport $D_{\hat{j}}$ to obtain a smaller loss.
This means $\pi(D_i^2,D_{\hat{j}}',D_{-i,-\hat{j}})=h_i$.
Now for fixed $D_{\hat{j}}',D_{-i,-\hat{j}}$, if for all $D_i'$ with $y_i^*=\median(D_i')$, $\pi(D_i',D_{\hat{j}}',D_{-i,-\hat{j}})=h_i$, specifically we have $\pi(D_i^1,D_{\hat{j}}',D_{-i,-\hat{j}})=h_i$.
However, when agent $\hat{j}$ reports $D_{\hat{j}}'$ and other agents report $D_{-i,-\hat{j}}$, agent $i$ with real locations $\bar{D_i}=D_i^1$ will misreport $D_{\hat{j}}$, which leads to $\pi(D_{\hat{j}},D_{\hat{j}}',D_{-i,-\hat{j}})=\pi(D_{\hat{j}},D_i^1,D_{-i,-\hat{j}})=s_i$ by the anonymity, to obtain a smaller loss.
This is a contradiction, meaning there must exist $s_i'<y_i^*$ such that for some $D_i'$ with $y_i^*=\median(D_i')$, $\pi(D_i',D_{\hat{j}}',D_{-i,-\hat{j}})=s_i'$.
If $s_i' < s_i$, then $y_{\hat{j}}^*\leq s_i' < s_i$, while $\pi(D_i',D_{\hat{j}},D_{-i,-\hat{j}})\geq s_i$ since it equals either $s_i$ or $h_i$.
This means when agent $i$ reports $D_i'$ and other agents report $D_{-i,-\hat{j}}$, agent $\hat{j}$ with real locations $\bar{D_{\hat{j}}}=D_{\hat{j}}$ will misreport $D_{\hat{j}}'$ to obtain smaller loss.
If $s_i' > s_i$, for the special location set $D_i^3$, we have $l(s_i',D_i^3) < l(s_i,D_i^3) = l(h_i,D_i^3)$.
Thus, $\pi(D_i^3,D_{\hat{j}}',D_{-i,-\hat{j}})=s_i'$, while originally $\pi(D_i^3,D_{\hat{j}},D_{-i,-\hat{j}})=h_i$.
\footnote{Note that if we assume $\pi(D_i^3,D_{-i})=s_i$, the proofs for $n\geq 3$ shall be done by an induction on the opposite index set $\mathcal{I}'=\{j|y_j^*=\median(D_j), y_j^*\geq h_i\}$.}
This means when agent $i$ reports $D_i^3$ and other agents report $D_{-i,-\hat{j}}$, agent $j$ with real location $\bar{D_{\hat{j}}}=D_{\hat{j}}$ will misreport $D_{\hat{j}}'$ to obtain smaller loss.
Thus, we must have $s_i'=s_i$.
In other words, when other agents' reports change to $D_{\hat{j}}',D_{-i,-\hat{j}}$, $\pi$ still satisfies Eqn. (\ref{app_eqn:charac_fl_odd}) with the same $s_i, h_i$ as they report  $D_{\hat{j}},D_{-i,-\hat{j}}$, while the corresponding $|\mathcal{I}| = t$ since $\median(D_{\hat{j}}')=y_i^*>s_i$.
However, we have already proved that it is impossible for $|\mathcal{I}| = t$.
By induction, we provide contradictions for $|\mathcal{I}|= 1,\dots, n-1$, and this completes the proof for $n\geq 3$.

\end{proof}

\begin{lemma}\label{app_lem:sp_eff_only_peak_even}
For facility location problems where each agent control multiple locations, if a mechanism $\pi$ is anonymous, efficient and strategyproof w.r.t. manipulation, and $m$ is an even number, then for any $i\in [n]$, $D_{-i}$, $y_i^0\leq y_i^1\in\mathbb{R}$, there do not exist $s_i,h_i\in\mathbb{R}$ such that $s_i< y_i^0\leq y_i^1< h_i$, and $\pi$ satisfies Eqn. (\ref{app_eqn:charac_fl_even}) for all $D_i$ with $[y_i^0,y_i^1]=\median(D_i)$.
\end{lemma}

The proof of Lemma \ref{app_lem:sp_eff_only_peak_even} makes use of Lemma \ref{app_lem:even}.

\begin{lemma}\label{app_lem:even}
For facility location problems where each agent controls multiple locations, suppose a mechanism $\pi$ is anonymous, efficient and strategyproof w.r.t. manipulation, and $m$ is an even number.
If there exist some $i\in [n]$, $D_{-i}$, $y_i^0\leq y_i^1\in\mathbb{R}$, $s_i,h_i\in\mathbb{R}$ such that $s_i< y_i^0\leq y_i^1< h_i$, and $\pi$ satisfies Eqn. (\ref{app_eqn:charac_fl_even}) for all $D_i$ with $[y_i^0,y_i^1]=\median(D_i)$, denote two index sets $\mathcal{I}_1=\{j|[y_j^0,y_j^1]=\median(D_j), y_j^0\leq s_i,y_j^1<h_i\}$ and $\mathcal{I}_2=\{j|[y_j^0,y_j^1]=\median(D_j), y_j^1\geq h_i,y_j^0>s_i\}$, then $\mathcal{I}_1=\emptyset,\mathcal{I}_2=\emptyset$.
\end{lemma}

\begin{proof}[Proof of Lemma \ref{app_lem:even}]
If $n=1$, an efficient mechanism should always output the median of the reported locations.

For $n=2$, w.o.l.g, assume there exists some $D_{2}$ with median denoted by $[y_2^0,y_2^1]$, $y_1^0,y_1^1$, and $s_1< y_1^0\leq y_1^1< h_1$, such that $\pi$ satisfies Eqn. (\ref{app_eqn:charac_fl_even}) for all $D_1$ with $[y_1^0,y_1^1]=\median(D_1)$.
This means there exist at least two special location sets $D_1^1$ and $D_1^2$ such that $\pi(D_1^1,D_2)=s_1$ and $\pi(D_1^2,D_2)=h_1$ (for example $D_1^1=\{s_1,y_1^0,y_1^1,y_1^1\}$ and $D_1^2=\{y_1^0,y_1^0,y_1^1,h_1\}$).
If $\mathcal{I}_1| \neq \emptyset$, that is $y_2^0\leq s_i$ and $y_2^1 < h_i$, then $\pi(D_1^2,D_2)=h_1$ is not efficient since any location between $y_1^0$ and $y_1^1$ is a better output for both agent $1$ and $2$.
If $\mathcal{I}_2 \neq \emptyset$, that is $y_2^1\geq h_i$ and $y_2^0 > s_i$, then $\pi(D_1^1,D_2)=s_1$ is not efficient since any location between $y_1^0$ and $y_1^1$ is a better output for both agent $1$ and $2$.
Thus, the result holds for $n=2$.

For $n\geq 3$, assume there exist some $i$, $D_{-i}$, $y_i^0\leq y_i^1$, and $s_i< y_i^0\leq y_i^1< h_i$, such that $\pi$ satisfies Eqn. (\ref{app_eqn:charac_fl_even}) for all $D_i$ with $[y_i^0,y_i^1]=\median(D_i)$.
Consider two special location sets $D_i^1,D_i^2$ satisfying $\pi(D_i^1,D_{-i})=s_i$, $\pi(D_i^2,D_{-i})=h_i$.
In the following proof, we shall derive contradictions for $|\mathcal{I}_1|= 1,\dots, n-1$ by induction, and the corresponding analysis for $\mathcal{I}_2$ is similar.

If $|\mathcal{I}_1|= 1$, that is there is only one $j\in [n]$, such that $[y_j^0,y_j^1]=\median(D_j)$, $y_j^0\leq s_i$ and $y_j^1 < h_i$.
Denote $D_{-i,-j}$ as the reported location sets by agents other than agent $i$ and $j$, then we rewrite the fact that $\pi(D_i^1,D_j,D_{-i,-j})=s_i$ and $\pi(D_i^2,D_j,D_{-i,-j})=h_i$.
Consider agent $j$'s another possible report $D_j'=D_i^1$.
Due to the efficiency, we know $\pi(D_i^1,D_j',D_{-i,-j}) > s_i \geq y_j^0$.
If $\pi(D_i^2,D_j',D_{-i,-j})>h_i$, then when agent $i$ reports $D_i^2$ and other agents report $D_{-i,-j}$, agent $j$ with real locations $\bar{D_j}=D_j'$ will misreport $D_j$ to obtain a smaller loss.
If $\pi(D_i^2,D_j',D_{-i,-j})<h_i$, then when agent $i$ reports $D_i^2$ and other agents report $D_{-i,-j}$, agent $j$ with real locations $\bar{D_j}=D_j$ will misreport $D_j'$ to obtain a smaller loss.
This means $\pi(D_i^2,D_j',D_{-i,-j})=h_i$.
Now for fixed $D_j',D_{-i,-j}$, suppose there exists $s_i'<y_i^0$ such that for some $D_i'$ with $[y_i^0,y_i^1]=\median(D_i')$, $\pi(D_i',D_j',D_{-i,-j})=s_i'$.
Let $s^*=\min_{k\in [n]}\{y_k^0|[y_k^0,y_k^1]=\median(D_k), s_i< y_k^0\leq y_k^1 < h_i\}$, and $k$ denote the corresponding agent with $y_k^0 = s^*$.
Due to the efficiency we must have $s_i' \geq s^*$ , which means there exists $k\neq i$ satisfying $y_k^0\leq s_i', y_k^1 < h_i$.
This means we can repeat previous process w.r.t. $D_k$ and new $[s_i',h_i]$ (for at most $n-2$ times since $k\neq i$ and $k\neq j$) until no such $k$.
We reuse the notation $D'_{-i}$ after this process (either with the repeated part or not) as the reports of agents other than $i$, and we have for all $D_i$ with $[y_i^0,y_i^1]=\median(D_i)$, $\pi(D_i,D_{-i}')=h_i$.
Specifically, $\pi(D_i^1,D_j',D_{-i,-j})=h_i$. (If the repeated process happens, we only need to replace $j$ and $D_i^1$ by the lass $k$ and $D_i'$ in the process.)
However, when agent $j$ reports $D_j'$ and other agents report $D_{-i,-j}$, agent $i$ with real locations $\bar{D_i}=D_i^1$ will misreport $D_j$, which leads to $\pi(D_j,D_j',D_{-i,-j})=\pi(D_j,D_i^1,D_{-i,-j})=\pi(D_i^1,D_j,D_{-i,-j})=s_i$ by its symmetry, to obtain a smaller loss.
This is a contradiction, meaning it is impossible that $|\mathcal{I}|= 1$.

Now suppose we have proved it is impossible that $|\mathcal{I}|= 1,\dots, t$ for some $t\geq 1$, we consider the case $|\mathcal{I}|= t+1$.
Let $\hat{j}=\mathop{\arg\min_{j\in\mathcal{I}_1}\{y_j^0|[y_j^0,y_j^1]=\median(D_j)\}}$.
Similarly, denote $D_{-i,-\hat{j}}$ as the reported location sets by agents other than agent $i$ and $\hat{j}$, and we rewrite the fact that $\pi(D_i^1,D_{\hat{j}},D_{-i,-\hat{j}})=s_i$ and $\pi(D_i^2,D_{\hat{j}},D_{-i,-\hat{j}})=h_i$.
Consider agent $\hat{j}$'s another possible report $D_{\hat{j}}'=D_i^1$.
If $\pi(D_i^1,D_{\hat{j}}',D_{-i,-j}) < y_{\hat{j}}^0$, it is not efficient since at least $y_{\hat{j}}^0$ is a better output for all agents $j\in\mathcal{I}_1$ as well as agent $i$ and not worse for other agents.
If $\pi(D_i^2,D_{\hat{j}}',D_{-i,-\hat{j}})>h_i$, then when agent $i$ reports $D_i^2$ and other agents report $D_{-i,-\hat{j}}$, agent $j$ with real locations $\bar{D_{\hat{j}}}=D_{\hat{j}}'$ will misreport $D_{\hat{j}}$ to obtain a smaller loss.
If $y_{\hat{j}}^0\leq \pi(D_i^2,D_{\hat{j}}',D_{-i,-\hat{j}})<h_i$, then when agent $i$ reports $D_i^2$ and other agents report $D_{-i,-\hat{j}}$, agent $\hat{j}$ with real locations $\bar{D_{\hat{j}}}=D_{\hat{j}}$ will misreport $D_{\hat{j}}'$ to obtain a smaller loss.
This means $\pi(D_i^2,D_{\hat{j}}',D_{-i,-\hat{j}})=h_i$.
Now for fixed $D_{\hat{j}}',D_{-i,-\hat{j}}$, if for all $D_i$ with $[y_i^0,y_i^1]=\median(D_i)$, $\pi(D_i,D_{\hat{j}}',D_{-i,-\hat{j}})=h_i$, specifically we have $\pi(D_i^1,D_{\hat{j}}',D_{-i,-\hat{j}})=h_i$.
Then when agent $\hat{j}$ reports $D_{\hat{j}}'$ and other agents report $D_{-i,-\hat{j}}$, agent $i$ with real locations $\bar{D_i}=D_i^1$ will misreport $D_{\hat{j}}$, which leads to $\pi(D_{\hat{j}},D_{\hat{j}}',D_{-i,-\hat{j}})=\pi(D_{\hat{j}},D_i^1,D_{-i,-\hat{j}})=s_i$ by its symmetry, to obtain a smaller loss.
This is a contradiction, meaning there must exist $s_i'<y_i^0$ such that for some $D_i'$ with $[y_i^0,y_i^1]=\median(D_i')$, $\pi(D_i',D_{\hat{j}}',D_{-i,-\hat{j}})=s_i'$.
In other words, when other agents' reports change to $D_{\hat{j}}',D_{-i,-\hat{j}}$, there still exist $s_i',h_i'$, such that $s_i' < y_i^0\leq y_i^1 < h_i'=h_i$ and $\pi$ satisfies Eqn. (\ref{app_eqn:charac_fl_odd}) for all $D_i$ satisfying that $|D_i|$ is an even number and $[y_i^0,y_i^1]=\median(D_i)$. 
And for the corresponding index set $\mathcal{I}_1' = \{j|[y_j^0,y_j^1]=\median(D_j), y_j^0\leq s_i',y_j^1<h_i\}$ (after replacing $D_{\hat{j}}$ by $D_{\hat{j}}'$), we have $\hat{j}\notin \mathcal{I}_1'$ since $\median(D_{\hat{j}}') = [y_i^0, y_i^1]$.
Finally, if $|\mathcal{I}_1'| > t$, we can repeatedly replace the report of the agent whose location set has the leftmost median among those agents in $\mathcal{I}_1'$ in the same way until $|\mathcal{I}_1'| \leq t$, and the same analysis still holds.
Since such an index set satisfies $|\mathcal{I}_1'| \leq n-1$, after at most $n-t-1$ times of such replacing, we must have $|\mathcal{I}_1'| = t$.
However, we have already proved that it is impossible for $|\mathcal{I}_1| = t$.
By induction, we provide contradictions for $|\mathcal{I}|= 1,\dots, n-1$, and this completes the proof for $n\geq 3$.

\end{proof}

\begin{proof}[Proof of Lemma \ref{app_lem:sp_eff_only_peak_even}]
Assume there exist some $i$, $D_{-i}$, $y_i^0\leq y_i^1$, and $s_i< y_i^0\leq y_i^1< h_i$, such that $\pi$ satisfies Eqn. (\ref{app_eqn:charac_fl_even}) for all $D_i$ satisfying that $|D_i|$ is an even number and $[y_i^0,y_i^1]=\median(D_i)$.
According to Lemma \ref{app_lem:even}, we know for any $j\in [n]$, let $[y_j^0,y_j^1]=\median(D_j)$, then either $y_j^0\leq s_i < h_i \leq y_j^1$, or $s_i< y_j^0\leq y_j^1 < h_i$.
Let $\mathcal{I}=\{j|[y_j^0,y_j^1]=\median(D_j), s_i< y_j^0\leq y_j^1 < h_i\}$, $\mathcal{I}\neq \emptyset$ since there exist at least one $j=i\in\mathcal{I}$.
Let $s^*=\min_{j\in\mathcal{I}}\{y_j^0|[y_j^0,y_j^1]=\median(D_j)\}$.
Then for any $D_i$ that $\pi(D_i,D_{-i})=s_i$ is not efficient since $s^*$ is a better output for all agent $j\in\mathcal{I}$ while not worse for other agent.
This is a contradiction to the efficiency condition and completes the proof.
\end{proof}

Lemma \ref{app_lem:sp_eff_only_peak_odd} and \ref{app_lem:sp_eff_only_peak_even} mean any mechanism $\pi$ which is efficient and strategyproof only depends on $y_i^*=\median(D_i)$ for $i\in[n]$.
In other words, agents only need to report their most preferred locations, i.e. $\median(D_i)$ for $i\in[n]$, which completes the proof of Theorem \ref{thm:simply_mani_charac}.

\end{document}